\documentclass[a4paper,UKenglish]{lipics-v2018}
\usepackage{microtype}%if unwanted, comment out or use option "draft"
\theoremstyle{theorem}
\newtheorem{conjecture}[theorem]{Conjecture}

\newtheorem{proposition}[theorem]{Proposition}

\newcommand{\comment}[1]{}
\title{Finding Short Synchronizing Words for Prefix Codes}

\titlerunning{Finding Short Synchronizing Words for Prefix Codes}%optional, please use if title is longer than one line

\author{Andrew Ryzhikov}{Universit\'e Paris-Est, LIGM, Marne-la-Vall\'ee, France}{ryzhikov.andrew@gmail.com}{}{}%{https://orcid.org/0000-0002-1825-0097}{[funding]}%mandatory, please use full name; only 1 author per \author macro; first two parameters are mandatory, other parameters can be empty.

\author{Marek Szyku\l a}{Institute of Computer Science, University of Wroc\l aw, Wroc\l aw,  Poland}{msz@cs.uni.wroc.pl}{}{Supported in part by the National Science Centre, Poland under project numbers 2017/25/B/ST6/01920 and 2014/15/B/ST6/00615}%{[orcid]}{[funding]}

\authorrunning{A. Ryzhikov and M. Szyku\l a}%mandatory. First: Use abbreviated first/middle names. Second (only in severe cases): Use first author plus 'et al.'

\Copyright{Andrew Ryzhikov and Marek Szyku\l a}%mandatory, please use full first names. LIPIcs license is "CC-BY";  http://creativecommons.org/licenses/by/3.0/

\subjclass{Theory of computation $\rightarrow$ Formal languages and automata}% mandatory: Please choose ACM 2012 classifications from https://www.acm.org/publications/class-2012 or https://dl.acm.org/ccs/ccs_flat.cfm . E.g., cite as "General and reference $\rightarrow$ General literature" or \ccsdesc[100]{General and reference~General literature}. 

%mandatory

\category{}%optional, e.g. invited paper

\relatedversion{}%optional, e.g. full version hosted on arXiv, HAL, or other respository/website

\supplement{}%optional, e.g. related research data, source code, ... hosted on a repository like zenodo, figshare, GitHub, ...

\funding{}%optional, to capture a funding statement, which applies to all authors. Please enter author specific funding statements as fifth argument of the \author macro.

\acknowledgements{The first author would like to thank Dominique Perrin for many useful discussions. We are also grateful to anonymous reviewers for their comments that improved presentation of the paper.}%optional

\EventEditors{Igor Potapov, Paul Spirakis, and James Worrell}
\EventNoEds{3}
\EventLongTitle{43rd International Symposium on Mathematical Foundations of Computer Science (MFCS 2018)}
\EventShortTitle{MFCS 2018}
\EventAcronym{MFCS}
\EventYear{2018}
\EventDate{August 27--31, 2018}
\EventLocation{Liverpool, GB}
\EventLogo{}
\SeriesVolume{117}
\ArticleNo{} % “New number” (=<article-no>) goes here!
\nolinenumbers %uncomment to disable line numbering
\hideLIPIcs  %uncomment to remove references to LIPIcs series (logo, DOI, ...), e.g. when preparing a pre-final version to be uploaded to arXiv or another public repository
\begin{document}

\maketitle

\begin{abstract}
We study the problems of finding a shortest synchronizing word and its length for a given prefix code. This is done in two different settings: when the code is defined by an arbitrary decoder recognizing its star and when the code is defined by its literal decoder (whose size
is polynomially equivalent to the total length of all words in the code). For the first case for every $\varepsilon > 0$ we prove $n^{1 - \varepsilon}$-inapproximability for recognizable binary maximal prefix codes, $\Theta(\log n)$-inapproximability for finite binary maximal prefix codes and $n^{\frac{1}{2} - \varepsilon}$-inapproximability for finite binary prefix codes. By $c$-inapproximability here we mean the non-existence of a $c$-approximation polynomial time algorithm under the assumption P $\ne$ NP, and by $n$ the number of states of the decoder in the input. For the second case, we propose approximation and exact algorithms and conjecture that for finite maximal prefix codes the problem can be solved in polynomial time. We also study the related problems of finding a shortest mortal and a shortest avoiding word.
\end{abstract}

\section{Introduction}

Prefix codes are a simple and powerful class of variable-length codes that are widely used in information compression and transmission. A famous example of prefix codes is Huffman's codes \cite{Huffman1952}. In general, variable length codes are not resistant to errors, since one deletion, insertion or change of a symbol can desynchronize the decoder causing incorrect decoding of all the remaining part of the message. However, in a large class of codes called synchronizing codes resynchronization of the decoder is possible in such situations. It is known that almost all maximal finite prefix codes are synchronizing \cite{Freiling2003}. Synchronization of finite prefix codes has been investigated a lot \cite{Berlinkov2016,  Biskup2008Thesis, Biskup2009, Capocelli1992, Schutzenberger1964, Schutzenberger1967}, see also the book \cite{Berstel2010} and references therein. For efficiency reasons, it is important to use as short words resynchronizing the decoder as possible to decrease synchronization time. However, despite the interest in synchronizing prefix codes, the computational complexity of finding short synchronizing words for them has not been studied so far. In this paper, we provide a systematic investigation of this topic.

Each recognizable (by a finite automaton) maximal prefix code can be represented by an automaton decoding the star of this code. For a finite code, this automaton can be exponentially smaller than the representation of the code by listing all its words (consider, for example, the code of all words of some fixed length). This can, of course, happen even if the code is synchronizing. An important example here is the code $0\{0, 1\}^{n - 1} \cup 1\{0, 1\}^n$. The minimized decoder of this code is a famous Wielandt automaton with $n + 1$ states (see ex. \cite{Ananichev2010}), while the literal automaton contains $2^{n - 1} + 2^n$ states, see Figure \ref{fig-wielandt} for the case $n = 3$. In different applications, the first or the second way of representing the code can be useful. In some cases large codes having a short description may be represented by a minimized decoder, while in other applications the code can be described by simply providing the list of all codewords. The number of states of the literal decoder is equal to the number of different prefixes of the codewords, and thus the representations of a prefix code by listing all its codewords and by providing its literal automaton are polynomially equivalent. We study the complexity of problems for both arbitrary and literal decoders of finite prefix codes.

\begin{figure}[hbt]
\includegraphics{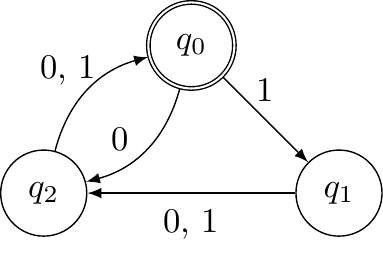}
\includegraphics{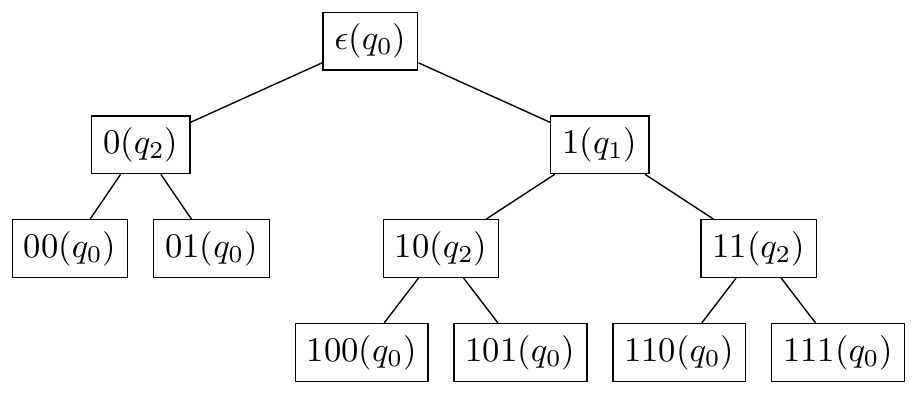}
 \caption{The Wielandt automaton on three states and the tree of the code $0\{0, 1\} \cup 1\{0, 1\}^2$.}\label{fig-wielandt}
\end{figure}

In this paper we study the existence of approximation algorithms for the problem {\sc Short Sync Word} of finding a shortest synchronizing words in several classes of deterministic automata decoding prefix codes. In Section \ref{sect-defs} we describe main definitions and survey existing results in the computational complexity of {\sc Short Sync Word}. In Section \ref{sect-sc} we provide a strong inapproximability result for this problem in the class of strongly connected automata. Section \ref{sect-acyclic} is devoted to the same problem in acyclic automata, which are then used in Section \ref{sect-decoders} to show logarithmic inapproximability of {\sc Short Sync Word} in the class of Huffman decoders. In Section \ref{sect-partial} we provide a much stronger inapproximability result for partial Huffman decoders. In Section \ref{sect-literal} we provide several algorithms for literal Huffman decoders and conjecture that {\sc Short Sync Word} can be solved in polynomial time in this class. Finally, in Section \ref{sect-mortal} we apply the developed techniques to the problems of finding shortest mortal and avoiding words.

\section{Main Definitions and Related Results} \label{sect-defs}

A {\em partial deterministic finite automaton} (which we simply call a {\em partial automaton} in this paper) is a triple $A = (Q, \Sigma, \delta)$, where $Q$ is a set of states, $\Sigma$ is a finite alphabet and $\delta: Q \times \Sigma \to Q$ is a (possibly incomplete) transition function. The function delta can be canonically extended to a function $\delta: Q \times \Sigma^* \to Q$ by defining $\delta(q, wx) = \delta(\delta(q, w), x)$ for $x \in \Sigma$, $w \in \Sigma^*$. If $\delta$ is a complete function, the automaton is called {\em complete} (in this case we call it just an {\em automaton}). An automaton is called {\em strongly connected} if for every ordered pair $q, q'$ of states there is a word mapping $q$ to $q'$.

A state in a partial automaton is called {\em sink} if each letter either maps the state to itself or is undefined. A {\em simple cycle} in a partial automaton $A = (Q, \Sigma, \delta)$ is a sequence $q_1, \ldots, q_k$ of its states such that all the states in the sequence are different and there exist letters $x_1, \ldots, x_k \in \Sigma$ such that $\delta(q_i, x_i) = q_{i + 1}$ for $1 \le i \le k - 1$ and $\delta(q_k, x_k) = q_1$. A simple cycle is a {\em self-loop} if it consists of only one state. We call a partial automaton {\em weakly acyclic} if all its cycles are self-loops, and {\em strongly acyclic} if moreover all its states with self-loops are sink states. Some properties of these automata have been studied in \cite{Ryzhikov2017}.

There is a strong relation between partial automata and prefix codes \cite{Berstel2010}. A set $X$ of words is called a {\em prefix code} if no word in $X$ is a prefix of another word. The class of recognizable (by an automaton) prefix codes can be described as follows. Take a strongly connected partial automaton $A$ and pick a state $r$ in it. Then the set of all {\em first return words} of $r$ (that is, words mapping $r$ to itself such that each non-empty prefix does not map $r$ to itself) is a recognizable prefix code. Moreover, each recognizable prefix code can be obtained this way. A prefix code is called {\em maximal} if it is not a subset of another prefix code. The class of maximal recognizable prefix codes corresponds to the class of complete automata. If a state $r$ can be picked in an automaton in such a way that the set of all first return words is a finite prefix code, we call the automaton a {\em partial Huffman decoder}. If such automaton is complete (and thus the finite prefix code is maximal), we call it simply a {\em Huffman decoder}.

Let $A$ be a partial automaton. A word $w$ is called {\em synchronizing} for $A$ if there exists a state $q$ such that $w$ maps each state of $A$ either to $q$ or the mapping of $w$ is undefined for this state, and there is at least one state such that the mapping of $w$ is defined for it. That is, a word is called synchronizing if it maps the whole set of states of the automaton to a set of size exactly one. An automaton having a synchronizing word is called {\em synchronizing}. A recognizable prefix code is {\em synchronizing} if a trim (partial) automaton recognizing the star of this code is synchronizing \cite{Berstel2010} (an automaton is called {\em trim} if there exists a state such that each state is accessible from this state, and there exists a state such that each state is coaccessible from this state). It can be checked in polynomial time that a strongly connected partial automaton is synchronizing (Proposition 3.6.5 of \cite{Berstel2010}).

Synchronizing automata have applications in different domains, such as synchronizing codes, symbolic dynamics, manufacturing and testing of reactive systems. They are also the subject of the \v{C}ern\'{y} conjecture, one of the main open problems in automata theory. It stays that every $n$-state synchronizing automaton has a synchronizing word of length at most $(n - 1)^2$, while the best known upper bounds are cubic \cite{Pin19832, Szykula2018Cerny}. See \cite{Volkov2008} for a survey on this topic. The upper bound on the length of a shortest synchronizing word has been improved in particular for Huffman decoders \cite{Beal2011} and further for literal Huffman decoders \cite{Berlinkov2016}.

We consider the following computational problem.

\begin{tabular}{||p{30em}}
	~{\sc Short Sync Word}\\
	~{\em Input}: A synchronizing partial automaton $A$;\\
	~{\em Output}: The length of a shortest synchronizing word for $A$.
\end{tabular}

Now we shortly survey existing results and techniques in the computational complexity and approximability of finding shortest synchronizing words for deterministic automata. To the best of our knowledge, there are no such results for partial automata. See \cite{Sipser2006} for an introduction to NP-completeness and \cite{Vazirani2001} for an introduction to inapproximability and gap-preserving reductions.

There exist several techniques of proving that {\sc Short Sync Word} is hard for different classes of automata. The very first and the most widely used idea is the one of Eppstein \cite{Eppstein1990}. Here, the automaton in the reduction is composed of a set of ``pipes'', and transitions define the way the active states are changed inside the pipes to reach the state where synchronization takes place. This idea (sometimes extended a lot) allows to prove NP-completeness of {\sc Short Sync Word} in the classes of strongly acyclic \cite{Eppstein1990}, ternary Eulerian \cite{Martyugin2012}, binary Eulerian \cite{Vorel2017Eul}, binary cyclic \cite{Martyugin2012} automata. This idea is also used in the proofs of \cite{Berlinkov2014Const} for inapproximability within arbitrary constant factor for binary automata, and for $n^{1 - \varepsilon}$-inapproximability for $n$-state binary automata \cite{Gawrychowski2015} (the last proof uses the theory of Probabilistically Checkable Proofs). In fact, the proof in \cite{Gawrychowski2015} holds true for binary automata with linear (in the number of states of the automaton) length of a shortest synchronizing word and a sink state.

Another idea is to construct a reduction from the {\sc Set Cover} problem. It can be used  to show logarithmic inapproximability of the {\sc Short Sync Word} in weakly acyclic \cite{Gerbush2011} and binary automata \cite{Berlinkov2014}. Finally, a reduction from {\sc Shortest Common Supersequence} provides inapproximability of this problem within a constant factor \cite{Gerbush2011}.

In the class of monotonic automata {\sc Short Sync Word} is solvable in polynomial time: because of the structure of these automata this problem reduces to a problem of finding a shortest words synchronizing a pair of states \cite{Ryzhikov2017Monotonic}. For general $n$-state automata, a $\lceil \frac{n - 1}{k - 1} \rceil$-approximation polynomial time algorithm exists for every $k$ \cite{Gerbush2011}.

\section{The Construction of Gawrychowski and Straszak} \label{sect-sc}

In this section we briefly recall the construction of a gadget invented by Gawrychowski and Straszak \cite{Gawrychowski2015} to show $n^{1 - \varepsilon}$-inapproximability of the {\sc Short Sync Word} problem in the general class of automata. Below we will use this construction several times.

Suppose that we have a constraint satisfiablity problem (CSP) with $N$ variables and $M$ constraints such that each constraint is satisfied by at most $K$ assignments (see \cite{Gawrychowski2015} for the definitions and missing details). Following the results in \cite{Gawrychowski2015}, we can assume that $N, K \le M^{\varepsilon}$, and also that either the CSP is satisfiable, or at most $\frac{1}{M^{1 - \varepsilon}}$ fraction of all constraints can be satisfied by an assignment. It is possible to construct the following ternary automaton $A_{\phi}$ in polynomial time. For each constraint $C$ the automaton $A_{\phi}$ contains a corresponding binary (over $\{0, 1\}$) gadget $T_C$ which is a compressed tree (that is, an acyclic digraph) of height $N$ and the number of states at most $N^2 K$ having different leaves corresponding to satisfying and non-satisfying assignments. The automaton $A_{\phi}$ also contains a sink state $s$ such that all the leaves corresponding to satisfying assignments are mapped to $s$, and all other leaves are mapped to the roots of the corresponding trees. The third letter is defined to map all the states of each gadget to its root and to map $s$ to itself. For every $\varepsilon > 0$ it is possible to construct such an automaton with at most $MN^2K \le M^{1 + 3\varepsilon}$ states in polynomial time. Moreover, for a satisfiable CSP we get an automaton with a shortest synchronizing word of length at most $N + 1 = O(M^\varepsilon)$, and for a non-satisfiable CSP the length of a shortest synchronizing word is at least $NM^{1 - \varepsilon} \ge M^{1 - \varepsilon}$. Since $\varepsilon$ can be chosen arbitrary small, this provides a gap-preserving reduction with a gap of $M^{1 - \varepsilon}$.

The described construction can be modified to get the same inapproximability in the class of strongly connected automata.

\begin{theorem} \label{tm-sc}
	The {\sc Short Sync Word} problem cannot be approximated in polynomial time within a factor of $n^{1 - \varepsilon}$ for every $\varepsilon > 0$ for $n$-state binary strongly connected automata unless P = NP.
\end{theorem}
\begin{proof}
	Consider the automaton $A_{\phi}$ described above. Add a new letter $c$ that cyclically permutes the roots of all gadgets, maps $s$ to the root of one of the gadgets and acts as a self-loop for all the remaining states. Observe that thus constructed automaton $A$ is strongly connected and has the property that every non-satisfying assignment satisfies at most the fraction of $\frac{1}{M^{1 - \varepsilon}}$ of all constraints. Thus, the gap between the length of a shortest synchronizing word for a satisfying and non-satisfying assignment is still $\Theta(M^{1 - \varepsilon})$. 
	
	It remains to make the automaton binary. This can be done by using Lemma 3 of \cite{Berlinkov2014}. This way we get a binary automaton with $\Theta(M^{1 + 3\varepsilon})$ states and a gap between $\Theta(M^\varepsilon)$ and $\Theta(M^{1 - \varepsilon})$ in the length of a shortest synchronizing word. By choosing appropriate small enough $\varepsilon$, we get a reduction with gap $n^{1 - \varepsilon}$ for binary strongly connected $n$-state automata, which proves the statement.
\end{proof}

\section{Acyclic Automata} \label{sect-acyclic}

In this section we investigate the simply-defined classes of weakly acyclic and strongly acyclic automata. The results for strongly acyclic automata are used in Section \ref{sect-decoders} to obtain inapproximability for Huffman decoders. Even though the automata in the classes of weakly and strongly acyclic automata are very restricted and have a very simple structure, the inapproximability bounds for them are quite strong. Thus we believe that these classes are of independent interest.

We will need the following problem.

\begin{tabular}{||p{32em}}
	~{\sc Set Cover} \\
	~{\em Input}: A set $X$ of $p$ elements and a family $C$ of $m$ subsets of $X$;\\
	~{\em Output}: A subfamily of $C$ of minimum size covering $X$.
\end{tabular}

A family $C'$ of subsets of $X$ is said to {\em cover} $X$ if $X$ is a subset of the union of the sets in $C'$. For every $\gamma > 0$, the {\sc Set Cover} problem with $|C| \le |X|^{1 + \gamma}$ cannot be approximated in polynomial time within a factor of $c' \log p$ for some $c' > 0$ unless P = NP \cite{Berlinkov2014}.  

\begin{theorem} \label{sa-set-cover}
	The {\sc Short Sync Word} problem cannot be approximated in polynomial time within a factor of $c \log n$ for some $c > 0$ for $n$-state strongly acyclic automata over an alphabet of size $n^{1 + \gamma}$ for every $\gamma > 0$ unless $P = NP$.
\end{theorem}
\begin{proof}
	We reduce the {\sc Set Cover} problem. Provided $X$ and $C$, we construct the automaton $A = (Q, \Sigma, \delta)$ as follows. To each set $c_k$ in $C$ we assign a letter $k \in \Sigma$. To each element $x_j$ in $X$ we assign a ``pipe'' of states $q^{(j)}_1, \ldots, q^{(j)}_p$ in $Q$. Additionally, we construct a state $f$ in $Q$.
	
	For $1 \le i \le p - 1$ and all $k$ and $j$ we define $\delta(q^{(j)}_i, k) = f$ if $c_k$ contains $x_j$, and $\delta(q^{(j)}_i, k) = q^{(j)}_{i + 1}$ otherwise. We also define  $\delta(q^{(j)}_p, k) = f$ for all $j$ and $k$.
	
	We claim that the length of a shortest synchronizing word for $A$ is equal to the minimum size of a set cover in $C$. Let $C'$ be a set cover of minimum size. Then a concatenation of the letters corresponding to the elements of $C'$ is a synchronizing word of corresponding length.
	
	In the other direction, consider a shortest synchronizing word $w$ for $A$. No letter appears in $w$ at least twice. If the length of $w$ is less than $p$, then by construction of $A$ the subset of elements in $C$ corresponding to the letters in $w$ form a set cover. Otherwise we can take an arbitrary subfamily of $C$ of size $p$ which is a set cover (such subfamily trivially exists if $C$ covers $X$).
	
	The resulting automaton has $p^2 + 1$ states and $m$ letters. Thus we get a reduction with gap $c' \log p \ge c'' \log \sqrt{|Q|} = \frac{1}{2} c'' \log |Q|$ for some $c'' > 0$. Because of the mentioned result of Berlinkov, we can also assume that $m < p^{1 + \gamma}$ for arbitrary small $\gamma > 0$. 
\end{proof}

Now we are going to extend this result to the case of binary weakly acyclic automata.

\begin{corollary}
	The {\sc Short Sync Word} problem cannot be approximated in polynomial time within a factor of $c \log n$ for some $c > 0$ for $n$-state binary weakly acyclic automata unless $P = NP$.
\end{corollary}
\begin{proof}
	We extend the construction from the proof of Theorem \ref{sa-set-cover} by using Lemma 3 of \cite{Berlinkov2014}. If we start with a strongly acyclic automaton with $p$ states and $k$ letters, this results in a binary weakly acyclic automaton with $4pk$ states. Moreover, the length of a shortest word of the new automaton is between $\ell (\log k + 1)$ and $(\ell + 1) (\log k + 1)$, where $\ell$ is the length of a shortest word of the original automaton. Since we can assume $p < k < p^{1 + \gamma}$ for arbitrary small $\gamma > 0$, we have $\log n = \Theta(\log p)$, where $n$ is the number of states of the new automaton. Thus we get a gap of $\Theta(\log p) = \Theta(\log n)$.
\end{proof}

For ternary strongly acyclic automata it is possible to get $(2-\varepsilon)$-inapproximability.

\begin{theorem} \label{sa-two}
	The {\sc Short Sync Word} problem cannot be approximated in polynomial time within a factor of $2 - \varepsilon$ for every $\varepsilon > 0$ for $n$-state strongly acyclic automata over an alphabet of size three unless $P = NP$.
\end{theorem}

\section{Huffman Decoders} \label{sect-decoders}

We start with a statement relating strongly acyclic automata to Huffman decoders.

\begin{lemma} \label{sa-to-huff}
	Let $A$ be a synchronizing strongly acyclic automaton over an alphabet of size~$k$. Let $\ell$ be the length of a shortest synchronizing word for $A$. Then there exists a  Huffman decoder $A_H$ over an alphabet of size $k + 2$ with the same length of a shortest synchronizing word, and $A_H$ can be constructed in polynomial time.
\end{lemma}
\begin{proof}
	Provided a strongly acyclic automaton $A = (Q, \Sigma, \delta)$ we construct a Huffman decoder $A_H = (Q_H, \Sigma_H, \delta_H)$.
	
	Since $A$ is a synchronizing strongly acyclic automaton, it has a unique sink state $f$. We define the alphabet $\Sigma_H$ as the union of $\Sigma$ with two additional letters $b_1, b_2$. The set of states $Q_H$ is the union of $Q$ with some auxiliary states defined as follows. Consider the set $S$ of states in $A$ having no incoming transitions. Construct a full binary tree with the root $f$ having $S$ as the set of its leaves (if $|S|$ is not a power of two, some subtrees of the tree can be merged). Define $b_1$ to map each state of this tree to the left child, and $b_2$ to the right child. Transfer the action of $\delta$ to $\delta_H$ for all states in $Q$ and all letters in $\Sigma$. For all the internal states of the tree define all the letters of $\Sigma$ to map these states to $f$. Finally, for all the states in $Q$ define the action of $b_1, b_2$ in the same way as some fixed letter in $\Sigma$.
	
	Observe that any word $w$ over alphabet $\Sigma$ synchronizing $A$ also synchronizes $A_H$. In the other direction, any synchronizing word for $A_H$ has to synchronize $Q$, which means that each state in $Q$ has to be mapped to $f$ first, so the length of a shortest synchronizing word for $A_H$ is at least the length of a shortest synchronizing word for $A$.
\end{proof}

Now we use Lemma \ref{sa-to-huff} to get preliminary inapproximability results for Huffman decoders.

\begin{corollary}\label{codes-weak}
	(i) The {\sc Short Sync Word} problem is NP-complete for Huffman decoders over an alphabet of size $4$.
	
	(ii) The {\sc Short Sync Word} problem cannot be approximated in polynomial time within a factor of $2 - \varepsilon$ for every $\varepsilon > 0$ for Huffman decoders over an alphabet of size $5$ unless $P = NP$.
	
	(iii) For every $\gamma > 0$, the {\sc Short Sync Word} problem cannot be approximated in polynomial time within a factor of $c \log n$ for some $c > 0$ for Huffman decoders over an alphabet of size $n^{1 + \gamma}$ unless $P = NP$.
\end{corollary}
\begin{proof}
	(i) The automaton in the Eppstein's proof of NP-completeness of {\sc Short Sync Word} \cite{Eppstein1990} is strongly acyclic. Then the reduction described in Lemma \ref{sa-to-huff} can be applied.
	
	(ii) A direct consequence of Theorem \ref{sa-two} and Lemma \ref{sa-to-huff}.
	
	(iii) A direct consequence of Theorem \ref{sa-set-cover} and Lemma \ref{sa-to-huff}.
\end{proof}

Now we show how to get a better inapproximability result for binary Huffman decoders using the composition of synchronizing prefix codes. We present a more general result for the composition of synchronizing codes which is of its own interest. This result shows how to change the size of the alphabet of a synchronizing complete code in such a way that the approximate length of a shortest synchronizing pair for it is preserved.

A set $X$ of words over an alphabet $\Sigma$ is a {\em code} if no word can be represented as a concatenation of elements in $X$ in two different ways. In particular, every prefix code is a code. A pair $(\ell_X, r_X)$ of words in $X^*$ is called {\em absorbing} if $\ell_X \Sigma^* \cap \Sigma^* r_X  \subseteq X^*$. The {\em length} of a pair is the total length of two word. A code $X$ over an alphabet $\Sigma$ is called {\em complete} if every word $w \in \Sigma^*$ is a factor of some word in $X^*$, that is, if for every word $w \in \Sigma^*$ there exist words $v_1, v_2 \in \Sigma^*$, $u \in X^*$ such that $v_1 w v_2 = u$. In particular, every maximal (by inclusion) code is complete. A complete code having an absorbing pair is called {\em synchronizing}. We refer to \cite{Berstel2010} for a survey on the theory of codes.

Let $Y$ be a code over $\Sigma_Y$ and $Z$ be a code over $\Sigma_Z$. Suppose that there exists a bijection $\beta: \Sigma_Y \to Z$. The {\em composition} $Y \circ_\beta Z$ is then defined as the code $X = \{\beta(y) \mid y \in Y\}$ over the alphabet $\Sigma_Z$ \cite{Berstel2010}. Here $\beta(y)$ is defined as $\beta(y_1)\beta(y_2)\ldots\beta(y_k)$ for $y = y_1 y_2 \ldots y_k$, $y_i \in \Sigma_Y$. Sometimes $\beta$ is omitted in the notation of composition.

\begin{theorem} \label{thm-composition}
	Let $Y$ and $Z$ be two synchronizing complete codes, such that $Z$ is finite and $m$ and $M$ are the lengths of a shortest and a longest codeword in $Z$. Suppose that the composition $Y \circ Z$ is defined. Then the code $X = Y \circ Z$ is synchronizing, and the length of a shortest absorbing pair for $X$ is between $m \ell$ and $2M \ell + 2c$, where $\ell$ is the length of a shortest absorbing pair for $Y$ and $c$ is the length of a shortest absorbing pair for $Z$. 
\end{theorem}
\begin{proof}
	Let $Y \subseteq \Sigma_Y^*$, $X, Z \subseteq \Sigma_Z^*$, and $\beta: \Sigma_Y \to Z$ be such that $X = Y \circ_\beta Z$. First, assume that $Y$ and $Z$ are synchronizing, and let $(\ell_Y, r_Y)$, $(\ell_Z, r_Z)$ be shortest absorbing pairs for $Y$ and $Z$. Then $\ell_Y \Sigma_Y^* \cap \Sigma_Y^* r_Y  \subseteq Y^*$ and $\ell_Z \Sigma_Z^* \cap \Sigma_Z^* r_Z \subseteq Z^*$. We will show that $p_1 = (\beta(\ell_Y) \ell_Z r_Z \beta(r_Y), \beta(\ell_Y) \ell_Z r_Z \beta(r_Y))$ is an absorbing pair for $X$. Consider the set $\beta(\ell_Y) \ell_Z r_Z \beta(r_Y) \Sigma_Z^* \cap \Sigma_Z^* \beta(\ell_Y) \ell_Z r_Z \beta(r_Y)$. It is a subset of the set $\beta(\ell_Y) (\ell_Z \Sigma_Z^* \cap \Sigma_Z^* r_Z) \beta(r_Y) \subseteq \beta(\ell_Y) Z^*  \beta(r_Y) = \beta(\ell_Y \Sigma_Y^* r_Y) \subseteq \beta(Y^*) = X^*$. Thus, $p_1$ is an absorbing pair for $X$. Moreover, the length of this pair is between $2m (|\ell_Y| + |r_Y|) + 2(|\ell_Z| + |r_Z|)$ and $2M (|\ell_Y| + |r_Y|) + 2(|\ell_Z| + |r_Z|)$.
	
	Conversely, assume that $(\ell_X, r_X)$ is a shortest absorbing pair for $X$, hence $\ell_X \Sigma_Z^* \cap \Sigma_Z^* r_X \subseteq X^*$. Then by the definition of composition $X^* \subseteq Z^*$ and $\ell_X, r_X \in Z^*$; thus, $(\ell_X, r_X)$ is also absorbing for $Z$. Next, let $\ell_Y = \beta^{-1}(\ell_X)$, $r_Y = \beta^{-1}(r_X)$, $\ell_Y, r_Y \in Y^*$. Then $\beta(\ell_Y \Sigma_Y^* \cap \Sigma_Y^* r_Y) = \ell_X Z^* \cap Z^* r_X \subseteq \ell_X \Sigma_Z^* \cap \Sigma_Z^* r_X \subseteq X^* = \beta(Y^*)$. Since the mapping $\beta$ is injective,  $\ell_Y B^* \cap B^* r_Y \subseteq Y^*$. Consequently $Y$ is synchronizing, and $(\ell_Y, r_Y)$ is an absorbing pair for it of length between $\frac{1}{M} (|\ell_X| + |r_X|)$ and $\frac{1}{m} (|\ell_X| + |r_X|)$.
	
	Summarizing, we get that the length of a shortest absorbing pair for $X$ is between $m (|\ell_Y| + |r_Y|)$ and $2M (|\ell_Y| + |r_Y|) + 2(|\ell_Z| + |r_Z|)$.
\end{proof}

\begin{figure}[!h]
	\centering
	\includegraphics{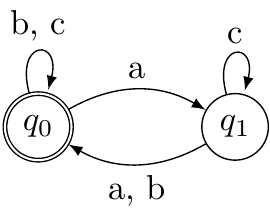}
	\hspace{0.3cm}
	\includegraphics{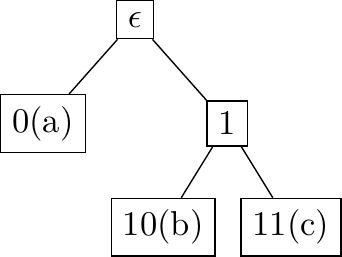}
	\includegraphics{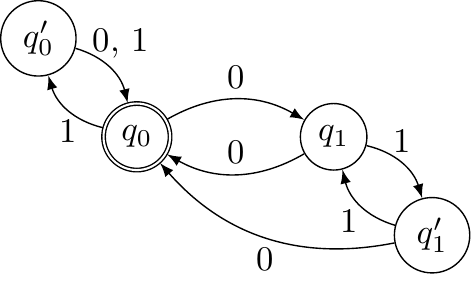}
	\caption{An automaton recognizing some infinite maximal prefix code, the tree of a finite maximal prefix code and an automaton recognizing their composition.}\label{fig-composition}
\end{figure}

In the case of maximal prefix codes the first element of the absorbing pair can be taken as an empty word. For recognizable maximal prefix codes $Y$ and $Z$, where $Z$ is finite, a Huffman decoder recognizing the star of $X = Y \circ Z$ can be constructed as follows. Let $H_Y$ be a Huffman decoder for $Y$. Consider the full tree $T_Z$ for $Z$, where each edge is marked by the corresponding letter. For each state $q$ in $H_Y$ we substitute the transitions going from this state with a copy of $T_Z$ as follows. The root of $T_Z$ coincides with $q$, and the inner vertices are new states of the resulting automaton. Suppose that $v$ is a leaf of $T_Z$, and the path from the root to $v$ is marked by a word $w$. Let $a$ be the letter of the alphabet of $H_Y$ which is mapped to the word $w$ in the composition. Then the image of $q$ under the mapping defined by $a$ is merged with $v$. In such a way we get a Huffman decoder with $\Theta(n_Y n_Z)$ states, where $n_Y, n_Z$ is the number of states in $H_Y$ and $T_Z$. By the definition of composition, this decoder has the same alphabet as $Z$. See Figure \ref{fig-composition} for an example.

\begin{corollary} \label{binary-huffman}
	The {\sc Short Sync Word} problem cannot be approximated in polynomial time within a factor of $c \log n$ for some $c > 0$ for binary $n$-states Huffman decoders unless $P = NP$. 
\end{corollary}
\begin{proof}
	We start with claim (iii) in Corollary \ref{codes-weak} and use Theorem \ref{thm-composition} to reduce the size of the alphabet. Thus, we reduce {\sc Short Sync Word} for Huffman decoders over an alphabet of size $n^{1 + \gamma}$ to {\sc Short Sync Word} for binary Huffman decoders.
	
	Assume that the size of the alphabet $k = n^{1 + \gamma}$ is a power of two (if no, duplicate some letter the required number of times). We take the code $0\{0, 1\}^{\log k - 1} \cup 1\{0, 1\}^{\log k}$ as $Z$. This is a code where some words are of length $\log k$ and the other words are of length $\log k + 1$ (after minimization the star of this code is recognized by a Wielandt automaton with $\log k + 1$ states discussed in the introduction). This code has a synchronizing word of length $\Theta((\log k)^2)$ \cite{Ananichev2010}. The number of vertices in the tree of this code is $\Theta(n)$.
	
	Let $\ell$ be the length of a shortest synchronizing word for the original automaton. By Theorem \ref{thm-composition}, the length of a shortest synchronizing word for the result of the composition is between $\ell \log k$ and $2(\log k + 1) \ell + \Theta((\log k)^2) = \Theta((\ell + \log k) \log k)$.
	
	For the {\sc Set Cover} problem the inapproximability result holds even if we assume that the size of the optimal solution is of size at least $d \log |X|$ for some $d > 0$. Indeed, if $d$ is a constant we can check all the subsets of $C$ of size at most $d \log |X|$ in polynomial time. Thus, we can assume that $\ell \ge \log k$ implying $(\ell + \log k) \log k  = \Theta(\ell \log k)$. Hence after the composition the length of a shortest synchronizing word is changed by at most constant multiplicative factor, and we we get a gap-reserving reduction with gap $\Theta(\log n)$. The resulting automaton is of size $\Theta(n^{2 + \gamma})$, and the $(2 + \gamma)$ dependence is hidden in the constant $c$ in the statement of this corollary.
\end{proof}

\section{Partial Huffman Decoders} \label{sect-partial}

In this section we investigate automata recognizing the star of a non-maximal finite prefix code. Such codes have some noticeable properties which do not hold for maximal finite prefix codes. For example, there exist non-trivial non-maximal finite prefix codes with finite synchronization delay, which provides guarantees on the synchronization time \cite{Bruyere1998}. This allows to read a stream of correctly transmitted compressed data from arbitrary position, which can be useful for audio and video decompression.

First we show that the known upper bounds and approximability for {\sc Short Sync Word} hold true for strongly connected partial automata. Because of Proposition 3.6.5 of \cite{Berstel2010}, Algorithm 1 of \cite{Volkov2008} works without any changes for strongly connected partial automata. The analysis of its approximation ratio is the same as in \cite{Gerbush2011}. Thus we get the following.

\begin{theorem}
	There exists a polynomial time algorithm (Algorithm 1 of \cite{Volkov2008}) finding a synchronizing word of length at most $\frac{n^3 - n}{6}$ for a $n$-state strongly connected partial automaton. Moreover, this algorithm provides a $O(n)$-approximation for the {\sc Short Sync Word} problem.
\end{theorem}

Now we provide a lower bound on the approximability of the {\sc Short Sync Word} problem for partial Huffman decoders by extending the idea used to prove inapproximability for Huffman decoders in the previous sections. First we prove the result for alphabet of size $5$ and then use a composition with a maximal finite prefix code to get the same result for the binary case.

\begin{theorem} \label{thm-partial-inapprox}
	The {\sc Short Sync Word} problem cannot be approximated within a factor of $n^{\frac{1}{2} - \varepsilon}$ for every $\varepsilon > 0$ for $n$-state partial Huffman decoders over an alphabet of size $5$ unless $P = NP$.
\end{theorem}
\begin{proof}
	First we prove inapproximability for the class of partial strongly acyclic automata, that is, automata having no simple cycles but loops in the sink state. We start with the CSP problem described in Section \ref{sect-sc} with all the restriction defined there. Having an instance of this problem with $N$ variables and $M$ constraints such that each constraint is satisfied by at most $K$ assignments, we construct an automaton $A^b_{\phi}$ over the alphabet $\{0, 1\}$. For each constraint $j$, we construct $M$ identical compressed trees $T_j^{1}, \ldots,  T_j^{M}$ corresponding to this constraint (also described in Section \ref{sect-sc}). Then for $1 \le i \le M - 1$ we merge the leaves of $T_j^{i}$ corresponding to non-satisfying assignments with the root of $T_j^{i + 1}$, and delete all the leaves corresponding to satisfying assignments (leaving all the transition leading to deleted states undefined). For each $T_j^{M}$, we again delete all the leaves corresponding to satisfying assignments and merge all the leaves corresponding to non-satisfying assignments with a new state $s$. This state is a self-loop, that is, $0, 1$ map $s$ to itself. Now we define an additional letter $a$ and $M$ new states $r_1, \ldots, r_M$. We define $a$ to map $r_j$ to the root of $T_j^1$. Finally, we add $N + 1$ new states $c_0, \ldots, c_N$ such that $a$ maps $c_0$ to $c_1$, and $0, 1$ map $c_i$ to $c_{i + 1}$ for $1 \le i \le N - 1$, and map $c_N$ to $s$. All other transitions are left undefined.
	
	If $a$ is applied first, the set $S$ of states to be synchronized is $c_1$ together with the roots of $T_j^{1}$ for all $j$. Observe that $a$ cannot be applied anymore, since it would result in mapping all the active states of the automaton to void. If a letter other than $a$ is applied first, a superset of $S$ must be synchronized then.
	
	If there exists a satisfying assignment $x_1, \ldots, x_N$ then the word $a x_1 \ldots x_N$ is synchronizing, since it maps all the states but $c_0$ to void. Otherwise, to synchronize the automaton we need to pass through $M$ compressed trees, since each tree can map only at most $M^\varepsilon$ states to void (since for every non-satisfiable CSP the maximum number of satisfiable constraints is $M^\varepsilon$ in the construction, see Section \ref{sect-sc}). Thus we get a gap of $M^{1 - \varepsilon} = n^{\frac{1}{2} - \varepsilon}$ for the class of $n$-state strongly acyclic partial automata.
	
	Now we are going to transfer this result to the case of partial Huffman decoders. We extend the idea of Lemma \ref{sa-to-huff}. All we need is to define transitions leading from $s$ to the states having no incoming transitions (which are $r_1, \ldots, r_M$ together with $c_0$). The only difference is that now we have to make sure that $a$ cannot be applied too early resulting in mapping all the states of the compressed trees to void leaving the state $s$ active.
	
	To do that, we introduce two new letters $b_1, b_2$ and perform branching as described in Lemma \ref{sa-to-huff}. Thus we get $M + 1$ leaves of the constructed full binary tree. To each leave we attach a chain of states of length $MN$ ending in the root of $T_j^{1}$ (or in $c_0$). This means that we introduce $MN$ new states and define the letters $b_1, b_2$ to map a state in each chain to the next state in the same chain. This guarantees that if the letter $a$ appears twice in a word of length at most $MN$, this word maps all the states of the automaton to void. Finally, the action of $b_1, b_2$ on the compressed trees and the states $c_0, \ldots, c_N$ repeats, for example, the action of the letter $0$.
	
	The number of states of the automaton in the construction is $O(M^{2 + 3 \varepsilon})$ The gap is then $M^{1 - 2 \varepsilon}$. By choosing small enough $\varepsilon$ we thus get a gap of $n^{\frac{1}{2} - \varepsilon}$ as required. 
\end{proof}

The next lemma shows that under some restrictions it is possible to reduce the alphabet of a non-maximal prefix code in a way that approximate length of a shortest synchronizing word is preserved. A word is called {\em non-mortal} for a prefix code $X$ if it is a factor of some word in $X^*$.

\begin{lemma} \label{prop-composition-partial}
	Let $Y, Z$ be synchronizing prefix codes such that $Z$ is finite and maximal. Let $m$ and $M$ be the lengths of a shortest and a longest codeword in $Z$. Suppose that $Y \circ_\beta Z$ is defined for some $\beta$. If there exists a synchronizing word $w_Z$ for $Z$ such that $\beta^{-1}(w)$ is a non-mortal word for $Y$, then the composition $X = Y \circ_\beta Z$ is synchronizing. Moreover, then the length of a shortest synchronizing word for $X$ is between $m\ell$ and $M \ell + |w_Z|$, where $\ell$ is the length of a shortest synchronizing word for $Y$.
\end{lemma}
\begin{proof}
	Let $Y \subseteq \Sigma_Y^*$, $X, Z \subseteq \Sigma_Z^*$, and $\beta: \Sigma_Y \to Z$ be such that $X = Y \circ_\beta Z$. Let $w_Y$ be a synchronizing word for $Y$. Then $w_Z \beta(w_Y)$ is a synchronizing word for $X$ of length at most $M \ell + |w_Z|$. In the other direction, let $w_X$ be a synchronizing word for $X$. Then $\beta^{-1}(w_X)$ is a synchronizing word for $Y$. Thus, $|w_X| \ge m\ell$.
\end{proof}

\begin{corollary}
	The {\sc Short Sync Word} problem cannot be approximated in polynomial time within a factor of $n^{\frac{1}{2} - \varepsilon}$ for every $\varepsilon > 0$ for binary $n$-state partial Huffman decoders unless $P = NP$.
\end{corollary}
\begin{proof}
	We use the composition of the automaton constructed in the proof of Theorem \ref{thm-partial-inapprox} with the prefix code $\{aaa, aab, ab, ba, bb\}$ having a synchronizing word $baab$. The word $baab$ is a concatenation of two different codewords, so their pre-images can be taken to be $a$ and $0$, resulting in a non-mortal word $a0$ for $A$, so we can use Lemma \ref{prop-composition-partial}.
\end{proof}

\section{Literal Huffman Decoders} \label{sect-literal}

In this section we deal with literal Huffman decoders. Given a finite maximal prefix code $X$ over an alphabet $\Sigma$, the literal automaton recognizing $X^*$ is an automaton $A = (Q, \Sigma, \delta)$ defined as follows.
The states of $A$ correspond to all proper prefixes of the words in $X$, and the transition function is defined as 

$$\delta(q, x) = \left\{
\begin{array}{ll}
qx &\mbox{if } qx \not \in X,\\
\epsilon &\mbox{if } qx \in X
\end{array}
\right.$$

We will need the following useful lemma. The {\em rank} of a word $w \in \Sigma^*$ with respect to an automaton $A = (Q, \Sigma, \delta)$ is the size of the image of $Q$ under the mapping defined by $w$.

\begin{lemma}[{\cite[Lemma~16]{Berlinkov2016}}] \label{lemma-logrank}
	For every $n$-state literal Huffman decoder over an alphabet of size $k$ there exists a word of length $\lceil \log_k n\rceil$ and rank at most $\lceil\log_k n\rceil$.
\end{lemma}

Note that if a $n$-state literal Huffman decoder has a synchronizing word of length at most $O(\log n)$, this word can be found in polynomial time by examining all words of length up to $O(\log n)$. Thus, in further algorithms we will assume that the length of a shortest synchronizing word is greater than this value. Lemma \ref{lemma-logrank} stays that a word of rank at most $\lceil \log_k n \rceil$ can also be found in polynomial time.

\begin{theorem} \label{thm-logapprox}
	There exists a $O(\log n)$-approximation polynomial time algorithm for the {\sc Short Sync Word} problem for literal Huffman decoders.
\end{theorem}
\begin{proof}
	Let $A = (Q, \Sigma, \delta)$ be a literal Huffman decoder, and $|\Sigma| = k$. Let $w$ be a word of rank at most $\lceil \log_k n \rceil$ found as described above.
	Let $Q'$ be the image of $Q$ under the mapping defined by $w$, i.e.\ $Q' = \delta(Q, w)$.
	Define by $v$ the word subsequently merging pairs of states in $Q'$ with shortest possible words.
	Note than a shortest word synchronizing $A$ has to synchronize every pair of states, in particular, one that requires a longest word.
	Thus the length of $v$ is at most $\lceil \log_k n \rceil$ times greater than the length of a shortest word synchronizing $A$.
	Then the word $wv$ is a $O(\log n)$-approximation for the {\sc Short Sync Word} problem.
\end{proof}

\begin{theorem} \label{superpoly}
	For every $\varepsilon > 0$, there exists a $(1 + \varepsilon)$-approximation $O(n^{\log n})$-time algorithm for the problem {\sc Short Sync Word} for $n$-state literal Huffman decoders.
\end{theorem}
\begin{proof}
	Let $A = (Q, \Sigma, \delta)$ be a literal Huffman decoder, and $|\Sigma| = k$. First we check all words of length at most $\frac{1}{\varepsilon} \lceil \log_k n \rceil$, whether they are synchronizing.
	The number of these words is polynomial, and the check can be performed in polynomial time.
	If a synchronizing word is found then we have an exact solution.
	Otherwise, a shortest synchronizing word must be longer than that and we proceed to the second stage.
	
	Let $w$ be a word of rank at most $\lceil \log_k n \rceil$ found as before.
	Now we construct the power automaton $A^{\le \lceil \log_k n \rceil}$ restricted to all the subsets of size at most $\lceil \log_k n \rceil$.
	Using it, we find a shortest word synchronizing the subset $\delta(Q,w)$; let this word be $v$.
	We return $wv$.
	
	Let $w'$ be a shortest synchronizing word for $|A|$.
	Clearly, $|w'| \ge |v|$ and $\varepsilon |w'| > \lceil \log_k n \rceil$.
	Thus $|wv| \le (1+\varepsilon)|w'|$, so $wv$ is a $(1 + \varepsilon)$-approximation as required.
\end{proof}

In view of the presented results we propose the following conjecture.

\begin{conjecture}
	There exists an exact polynomial time algorithm for the {\sc Short Sync Word} problem for literal Huffman decoders.
\end{conjecture}

Finally, we remark that it is possible to define the notion of the literal automaton of a non-maximal finite prefix code in the same way. In this case we leave undefined the transitions for a state $w$ and a letter $a$ such that $w$ is a proper prefix of a codeword, but $wa$ is neither a proper prefix of a codeword nor a codeword itself. However, the statement of Lemma \ref{lemma-logrank} is false for partial automata. Indeed, consider a two-word prefix code $\{(0^n1^n)^n, (1^n0^n)^n\}$. Its literal automaton has $2n^2 - 1$ states, and a shortest synchronizing word for it is $0^{n + 1}$ of length $n + 1$. Every word of length at most $n$ which is defined for at least one state is of the form $0^*1^*$ or $1^*0^*$ and thus has rank at least~$n - 1$.  

\section{Mortal and Avoiding Words} \label{sect-mortal}

A word $w$ is called {\em mortal} for a partial automaton $A$ if its mapping is undefined for all the states of $A$. The techniques described in this paper can be easily adapted to get the same inapproximability for the {\sc Short Mortal Word} problem defined as follows. 

\begin{tabular}{||p{32em}}
	~{\sc Short Mortal Word} \\
	~{\em Input}: A partial automaton $A$ with at least one undefined transition;\\
	~{\em Output}: The length of a shortest mortal word for $A$.
\end{tabular}

This problem is connected for instance to the famous Restivo's conjecture \cite{Restivo1981}. 

\begin{theorem}
	Unless P = NP, the {\sc Short Mortal Word} problem cannot be approximated in polynomial time within a factor of
	
	(i) $n^{1 - \varepsilon}$ for every $\varepsilon > 0$ for $n$-state binary strongly connected partial automata;
	
	(ii) $c \log n$ for some $c > 0$ for $n$-state binary partial Huffman decoders.
\end{theorem}
\begin{proof}
It can be seen that in Theorem \ref{tm-sc} and Corollary \ref{binary-huffman} we construct an automaton with a state $s$ such that each state has to visit $s$ before synchronization. Introduce a new state $s'$ having all the transitions the same as $s$, and for $s$ set the only defined transition (for an arbitrary letter) to map to $s'$. Thus we get an automaton such that every mortal word has to map each state to $s$ before mapping it to nowhere. Thus we preserve all the estimations on the length of a shortest mortal word, which proves both statements.
\end{proof}

Moreover, it is easy to get a $O(\log n)$-approximation polynomial time algorithm for {\sc Short Mortal Word} for literal Huffman decoders following the idea of Theorem \ref{thm-logapprox}. Indeed, it follows from Lemma \ref{lemma-logrank} that either there exists a mortal word of length at most $\lceil \log_k n \rceil$, or there exists a word $w$ of rank at most $\lceil \log_k n \rceil$. In the latter case we can find a word which is a concatenation of $w$ and a shortest word mapping all this states to nowhere one by one. By the arguments similar to the proof of Theorem \ref{thm-logapprox} we then get the following.

\begin{proposition}
	There exists a $O(\log n)$-approximation polynomial time algorithm for the {\sc Short Mortal Word} problem for $n$-state literal Huffman decoders. This algorithm always finds a mortal word of length $O(n \log n)$.
\end{proposition}

Another connected and important problem is to find a shortest avoiding word. Given an automaton $A = (Q, \Sigma, \delta)$, a word $w$ is called {\em avoiding} for a state $q \in Q$ if $q$ is not contained in the image of $Q$, that is, $q \not \in \delta(Q, w)$. Avoiding words play an important role in the recent improvement on the upper bound on the length of a shortest synchronizing word \cite{Szykula2018Cerny}. They are in some sense dual to synchronizing words.

\begin{tabular}{||p{32em}}
	~{\sc Short Avoiding word} \\
	~{\em Input}: An automaton $A$ and its state $q$ admitting a word avoiding $q$;\\
	~{\em Output}: The length of a shortest word avoiding $q$ in $A$.
\end{tabular}

If $q$ is not the root of a literal Huffman decoder $A$ (that is, not the state corresponding to the empty prefix), then a shortest avoiding word consists of just one letter. So avoiding is non-trivial only for the root state.

\begin{proposition}
	For every $\varepsilon > 0$, there exists a $(1 + \varepsilon)$-approximation $O(n^{\log n})$-time algorithm for the problem {\sc Short Sync Word} for $n$-state literal Huffman decoders.
\end{proposition}
\begin{proof}
We use the same algorithm as in the proof of Theorem \ref{superpoly}.
The only difference is that we check whether the words are avoiding instead of synchronizing.
\end{proof}

\section{Concluding Remarks}

For prefix codes, a synchronizing word is usually required to map all the states to the root \cite{Berstel2010}. One can see that this property holds for all the constructions of the paper. Moreover, in all the constructions the length of a shortest synchronizing word is linear in the number of states of the automaton. Thus, if we restrict to this case, we still get the same inapproximability results. Also, it should be noted that all the inapproximability results are proved by providing a gap-preserving reduction, thus proving NP-hardness of approximating the {\sc Short Sync Word} problem within a given factor.

\bibliography{SyncBib}
\end{document}